\documentclass[letterpaper, 10 pt, conference]{ieeeconf}  

\IEEEoverridecommandlockouts                              

\usepackage{graphics} 
\usepackage{epsfig} 
\usepackage{mathptmx} 
\usepackage{times} 
\usepackage{amsmath} 
\usepackage{amssymb}  
\usepackage{soul,color}

\usepackage{amsmath,amssymb,amsfonts,amstext}
\usepackage{graphicx}
\usepackage{textcomp}
\usepackage{babel}
\usepackage{balance}
\usepackage{comment}
\usepackage{todonotes}
\usepackage{mathtools}

\usepackage{amsthm}
\newtheorem{theorem}{Theorem}

\newtheorem{remark}{Remark}

\newtheorem{proposition}{Proposition}

\newtheorem{problem}{Problem}

\usepackage{algorithm}
\usepackage{algorithmicx} 
\usepackage{algpseudocode}

\makeatletter
\let\NAT@parse\undefined
\makeatother
\usepackage{hyperref}
\hypersetup{
    colorlinks=true,
    linkcolor=blue,
    allcolors=blue,
    citecolor=blue,
    filecolor=black,      
    urlcolor=blue,
}

\usepackage{tikz}
\usetikzlibrary{shapes, arrows.meta, positioning}
\usepackage{amsmath}
\usepackage{float} 
\usepackage{graphicx}
\usepackage{caption}
\usepackage{subcaption}
\usepackage{dblfloatfix}

\title{\LARGE \bf
A United Framework for Planning Electric Vehicle Charging Accessibility 
}

\author{Tony Kinchen$^{1}$, Panagiotis Typaldos$^{2}$, and Andreas A. Malikopoulos$^{3}$%
\thanks{*This research was supported in part by NSF under Grants CNS-2401007, CMMI-2348381, IIS-2415478, and in part by MathWorks.}%
\thanks{$^{1}$Tony Kinchen is with the Program of Systems Engineering, Cornell University, Ithaca, NY, USA. {\tt tjk238@cornell.edu}}%
\thanks{$^{2}$Panagiotis Typaldos is with the Department of Civil and Environmental Engineering, Cornell University, Ithaca, NY, USA. {\tt pt432@cornell.edu}}%
\thanks{$^{3}$Andreas A. Malikopoulos is with both the Program of Systems Engineering and the Department of Civil and Environmental Engineering, Cornell University, Ithaca, NY, USA. {\tt amaliko@cornell.edu}}%
}

\begin{document}

\maketitle
\thispagestyle{empty}
\pagestyle{empty}

\begin{abstract}
The shift towards electric vehicles (EVs) is crucial for establishing sustainable and low-emission urban transportation systems. However, the success of this transition depends on the strategic placement of the charging infrastructure. This paper addresses the challenge of optimizing charging station locations in dense urban environments while balancing efficiency with spatial accessibility. We propose an optimization framework that integrates traffic simulation, energy consumption modeling, and a mobility equity measure to evaluate the social reach of each potential charging station. Using New York City as a case study, we demonstrate consistent improvements in accessibility (15-20\% reduction in travel time variability). Our results provide a scalable methodology for incorporating equity considerations into EV infrastructure planning, although economic factors and grid integration remain important areas for future development.
\end{abstract}

\section{INTRODUCTION}

With increasing population and urbanization, traffic congestion has become a major problem for transportation systems worldwide. The increasing adoption of electric vehicles (EVs) introduces additional challenges related to the availability of materials, energy resources, and supporting infrastructure. With these changes in mind, strategic planning and effective deployment of charging stations become critical to meeting growing urban mobility demands while ensuring sustainability \cite{TowardsSustainableMobility}.

Despite the ongoing expansion toward EVs, the poor placement of charging stations often leads to bottlenecks and oversaturation in specific areas, resulting in inefficiencies and limited usability. To this end, infrastructure planning needs to prioritize both cost and efficiency, as well as spatial accessibility. This challenge is particularly pronounced in densely urban environments, where underserved communities often face limited access to essential transportation services.

\subsection{Related Work}

Recent research has explored various optimization techniques for EV charging station siting to improve utilization and system efficiency \cite{ImprovingMobility}. Approaches such as immune algorithms, Voronoi diagrams, and vehicle behavior modeling have demonstrated success in improving network performance and accessibility in urban environments \cite{EVChargingPlacement}. Environmental concerns are increasingly integrated into planning through methods like locational marginal emissions analyses, which enable infrastructure placement strategies that reduce carbon impacts while maintaining operational efficiency \cite{LocationalMarginalEmissions, sang2023lme}. Efforts to co-locate stations with renewable sources such as solar or wind have demonstrated benefits in minimizing grid reliance and environmental impacts \cite{GreenIoVEcosystem, OptimizationChargingStations}. Additionally, policy assessments such as the New York State Transportation Electrification Report and research on smart mobility highlight investment needs and opportunities for aligning infrastructure with sustainable development objectives \cite{NYSERDAElectrification, SmartMobilitySDGs}.

Simulation tools such as the Simulation of Urban MObility (SUMO) model traffic dynamics and EV charging behavior under realistic load conditions \cite{behrisch2011sumo, spachtholz2024charging, song2022sumo}. On a broader scale, system-level modeling frameworks analyze drivetrain configurations and energy efficiency in EVs \cite{SystemModelingHEV}. Eco-routing algorithms have also been proposed to reduce vehicle energy consumption \cite{QuantifyingEnvironmentalImpacts, asamer2016eco, 2024_Sagaama, DeNunzio2017}, while distributed EV-to-EV charging frameworks aim to improve resilience and flexibility in charging networks \cite{SoftwareDefinedEVtoEV}.

Despite these efforts, many of these strategies fail to directly address efficient accessibility. While proximity to demand centers and environmental objectives are considered, spatial accessibility is often overlooked. Recent advances in connected automated vehicles (CAVs) and coordinated traffic systems have shown potential for improving accessibility and network efficiency \cite{le2024distributed}. Techniques such as Bayesian optimization and dynamic trajectory coordination further enable real-time adaptation to evolving traffic patterns \cite{le2024controller, chalaki2020experimental}. Dynamic trip accessibility considerations are also increasingly relevant in EV infrastructure planning \cite{Bai2025routing, typaldos2025combining}. However, technoclogical improvements alone are not sufficient-ensuring accessibility requires metrics that are capable of accounting for how the success of the system is distributed across populations. 

\subsubsection{The Mobility Equity Metric (MEM)}

Efficient access to transportation infrastructure is essential for sustainable and inclusive urban development. Traditional transportation metrics often focus on efficiency—minimizing travel time or cost—without considering who benefits from these improvements. As a result, under-served areas may continue to face limited access to critical services such as jobs, healthcare, and education. MEM addresses this concern by providing a way to assess how fairly transportation resources and opportunities are distributed across different regions and populations \cite{Bang2024emergingequity}.

MEM evaluates both the ease of reaching essential services from different locations and the fairness of that access across the entire population. Rather than focusing on averages, it emphasizes whether all communities—regardless of income or geographic location—have similar levels of access. A high MEM value indicates that accessibility is distributed equitably, while a low value signals that some groups face significant barriers to mobility. The original formulation focused on generalized accessibility, but recent refinements have incorporated spatial reachability and economic cost to improve the metric’s resolution and computational tractability \cite{Bang2024ifac}.

This metric considers factors such as transportation mode availability, travel costs, proximity to important services, and the ability of residents to reach destinations within reasonable time frames. It also accounts for how differences in income or mobility constraints affect people’s real-world ability to access these services. In recent applications, MEM has been embedded into a game-theoretic routing framework, where it serves as an optimization objective to guide multi-modal traffic flows based on varying traveler behaviors \cite{Bang2024acc}.

By highlighting disparities in transportation access, MEM helps policymakers and planners design infrastructure projects that not only improve overall system performance but also advance social accessibility. In the context of electric vehicle charging infrastructure, this means ensuring that new stations are located in a way that benefits all communities—not just those in wealthier or better-connected areas. Such frameworks support resource allocation that balances societal fairness with efficiency considerations \cite{Malikopoulos2024}.

Additionally, the integration of real-time data from connected vehicles and the application of a stochastic control framework can further improve accessibility outcomes \cite{le2024stochastic}. These methods help guide vehicles toward available charging stations while minimizing disruptions and improving overall infrastructure resilience \cite{Bang2024emergingequity}.

\subsection{Contributions}



This paper presents a unified framework for planning equitable and efficient EV charging infrastructure in dense urban environments. Addressing the limitations of existing methods that often neglect spatial accessibility and system dynamics, we make the following specific contributions: (i) We develop the first framework to directly integrate the Mobility Equity Metric (MEM) into EV charging station optimization, providing a quantitative approach to balance efficiency with spatial accessibility; (ii) through large-scale SUMO simulations using NYC road networks, we demonstrate that accessibility-aware planning achieves measurable improvements in travel time equity (reduced variability) with modest distance trade-offs, providing empirical evidence for the viability of equity-focused infrastructure planning; (iii) we establish mathematical relationships that formally characterize the efficiency-accessibility trade-offs, enabling planners to make informed decisions about the weighting parameters based on policy priorities; and our optimization approach is computationally tractable for large urban networks and provides a foundation for incorporating equity considerations into existing infrastructure planning tools.

\subsection{Paper Structure}

The remainder of this paper is organized as follows: In section II we describe the problem formulation and optimization model. In section III we present the simulation and experimental setup. In section IV we evaluate system performance, and in Section V we conclude with future research directions and policy implications.

\section{Problem Formulation}

In our optimization framework, several key variables are introduced to model the spatial allocation and service dynamics of EV charging infrastructure. The set \( \mathrm{I} \) represents the locations of EV demand points, while \( \mathrm{J} \) denotes the candidate charging station selections. Each demand point \( \mathrm{i} \in \mathrm{I} \) is associated with a nonnegative weight \( \mathrm{d}_{\mathrm{i}} \), representing the intensity or frequency of charging demand at that location.The travel cost or distance between a demand point \( \mathrm{i} \) and a candidate charging station \( \mathrm{j} \) is denoted by \( \mathrm{c}_{\mathrm{i}\mathrm{j}} \geq 0 \). The binary decision variable \( \mathrm{x}_{\mathrm{i}\mathrm{j}} \in \{0,1\} \) indicates whether demand point \( \mathrm{i} \) is assigned to charging station \( \mathrm{j} \), and \( \mathrm{y}_{\mathrm{j}} \in \{0,1\} \) indicates whether a candidate station \( \mathrm{j} \) is selected . Each station location \( \mathrm{j} \in \mathrm{J} \) has a capacity limit \( \mathrm{s}_{\mathrm{j}} > 0 \), representing the maximum number of EVs it can serve. To account for accessibility considerations, each charging station \( \mathrm{j} \) is also assigned an accessibility score \( \varepsilon_{\mathrm{i}} \geq 0 \), and a tunable parameter \( \lambda \geq 0 \) is introduced to balance the trade-off between minimizing travel cost and improving accessiblity. The total number of stations allowed for selection is constrained by the parameter \( \mathrm{p} \in \mathbb{Z}_{\geq 0} \). A summary of these variables is provided in Table~\ref{tab:model_variables}.

\begin{table}[tb]
\centering
\caption{Summary of Optimization Model Variables}
\label{tab:model_variables}
\begin{tabular}{ll}
\hline
\textbf{Symbol} & \textbf{Description} \\
\hline
$ \mathrm{I} $ & Set of EV demand points \\
$ \mathrm{J} $ & Set of candidate charging station locations \\
$ \mathrm{d}_{\mathrm{i}} $ & Demand weight at point \( \mathrm{i} \in \mathrm{I} \) \\
$ \mathrm{c}_{\mathrm{i}\mathrm{j}} $ & Distance or travel time from \( \mathrm{i} \) to \( \mathrm{j} \) \\
$ \mathrm{x}_{\mathrm{i}\mathrm{j}} $ & Binary variable: 1 if point \( \mathrm{i} \) is assigned to station \( \mathrm{j} \) \\
$ \mathrm{y}_{\mathrm{j}} $ & Binary variable: 1 if station \( \mathrm{j} \) is selected \\
$ \mathrm{s}_{\mathrm{j}} $ & Capacity of station \( \mathrm{j} \) \\
$ \varepsilon_{\mathrm{i}} $ & Accessibility score for demand point \( \mathrm{i} \) \\
$ \lambda $ & Weighting parameter for accessibility \\
$ \mathrm{p} $ & Total number of stations \\
\hline
\end{tabular}
\end{table}

To effectively incorporate accessibility considerations into the optimization process, we integrate the MEM directly into our infrastructure planning framework. MEM quantitatively evaluates disparities in accessibility across a transportation network, helping planners identify where infrastructure investments can most effectively close access gaps. This is achieved by computing a MI for each location, which captures the ease of reaching essential services based on travel time, cost, and available modes.

Building on this, the MEM uses a Gini-style inequality measure to assess how evenly accessibility is distributed across the population. Higher MEM values indicate a more equitable system, while lower values highlight regions or populations facing significant mobility challenges. By integrating MEM into our optimization model, we prioritize infrastructure placements that not only improve efficiency but also promote greater access, ensuring that underserved communities receive greater consideration in planning decisions.

\textbf{Definition 1.}  
MI, denoted as \( \varepsilon_{\mathrm{i}} \), measures the level of accessibility to essential services from a specific location or node \( \mathrm{i} \in \mathrm{V} \). It accounts for travel costs, transportation modes, service importance, and time constraints. The MI is given as:
\begin{equation*}
    \varepsilon_{\mathrm{i}} = \sum_{\mathrm{m} \in \mathrm{M}} e^{-\kappa_{\mathrm{i}} \mathrm{c}_{\mathrm{m}}} \left( \sum_{\mathrm{s} \in \mathrm{S}} \beta_{\mathrm{s}} \, \tilde{\sigma}^{\mathrm{s}}_{\mathrm{i},\mathrm{m}}(\tau_{\mathrm{m}}) \right),
\end{equation*}
where \( \mathrm{M} \) represents the set of transportation modes (e.g., walking, public transit, private car, and cycling), \( \mathrm{S} \) is the set of service types (such as hospitals, grocery stores, pharmacies, and schools), \( \mathrm{c}_{\mathrm{m}} \) is the travel cost per mile for mode \( \mathrm{m} \), \( \kappa_{\mathrm{i}} \) captures price sensitivity at node \( \mathrm{i} \), \( \beta_{\mathrm{s}} \) reflects service priority (with higher weights for essential services like medical facilities and lower weights for non-essential ones like restaurants), and \( \tilde{\sigma}^{\mathrm{s}}_{\mathrm{i},\mathrm{m}}(\tau_{\mathrm{m}}) \) is the normalized number of accessible services of type \( \mathrm{s} \) within time threshold \( \tau_{\mathrm{m}} \).


\textbf{Definition 2.}  
The accessibility-weighted infrastructure objective quantifies the benefit of improving charging access across the network. It is defined as:
\[
\lambda \sum_{\mathrm{i} \in \mathrm{I}} \sum_{\mathrm{j} \in \mathrm{J}} \varepsilon_{\mathrm{i}} \, \mathrm{x}_{\mathrm{i}\mathrm{j}} \, \mathrm{y}_{\mathrm{j}},
\]
where \( \lambda \geq 0 \) is a tunable parameter controlling the importance of accessibility in the objective function, \( \varepsilon_i \geq 0 \) is the accessibility weight assigned to demand location \( i \in \mathrm{I} \), \( \mathrm{x}_{\mathrm{i}\mathrm{j}} \in \{0,1\} \) indicates whether demand point \( i \) is assigned to station \( j \), and \( \mathrm{y}_{\mathrm{j}} \in \{0,1\} \) indicates whether a candidate charging station is selected at location \( j \).

To formalize the practical and operational limitations involved in EV charging infrastructure deployment, a series of constraints are assigned to ensure feasible station placement, balanced demand distribution, and compliance with system capacity limits.

To ensure that each demand point is served by one and only one station, we impose the following constraint:
\begin{equation}
    \sum_{\text{j} \in \mathrm{J}} \mathrm{x}_{\text{i}\text{j}} = 1, \quad \forall\, \text{i} \in \mathrm{I}, \tag{1}
\end{equation}
where \( \mathrm{x}_{\text{i}\text{j}} = 1 \) only if demand point \( \text{i} \) is assigned to station \( \text{j} \), ensuring full coverage of all EV demand locations.

To limit the number of charging stations that can be used, we impose the following constraint:
\begin{equation}
    \sum_{\text{j} \in \mathrm{J}} \mathrm{y}_{\text{j}} = \mathrm{p}, \tag{2}
\end{equation}
where \( \mathrm{p} \in \mathbb{Z}_{\geq 0} \) is the total number of stations permitted in the infrastructure plan.

To ensure that demand is only assigned to operational stations, we impose the following consistency constraint:
\begin{equation}
    \mathrm{x}_{\text{i}\text{j}} \leq \mathrm{y}_{\text{j}}, \quad \forall\, \text{i} \in \mathrm{I},\ \forall\, \text{j} \in \mathrm{J}, \tag{3}
\end{equation}
where this condition links the assignment variable \( \mathrm{x}_{\text{i}\text{j}} \) to the station selection decision \( \mathrm{y}_{\text{j}} \), ensuring assignments are only made to selected stations.

To enforce the service capacity limitations of each charging station, we impose the following constraint:
\begin{equation}
    \sum_{\text{i} \in \mathrm{I}} \mathrm{x}_{\text{i}\text{j}} \leq \mathrm{s}_{\text{j}}\, \mathrm{y}_{\text{j}}, \quad \forall\, \text{j} \in \mathrm{J}, \tag{4}
\end{equation}
where \( \mathrm{s}_{\text{j}} \) is the maximum number of demand points that station \( \text{j} \) can support, ensuring stations are not overloaded.

To ensure that the selected stations provide sufficient capacity to meet overall system demand, we impose the following feasibility constraint:
\begin{equation}
    \sum_{\text{j} \in \mathrm{J}} \mathrm{s}_{\text{j}}\, \mathrm{y}_{\text{j}} \geq \sum_{\text{i} \in \mathrm{I}} \mathrm{d}_{\text{i}}, \tag{5}
\end{equation}
where \( \mathrm{s}_{\text{j}} \) is the capacity of station \( \text{j} \), and \( \mathrm{d}_{\text{i}} \) is the demand weight at location \( \text{i} \). This condition guarantees that the combined capacity of all selected stations is at least equal to the total demand across the network, preventing infeasible configurations.


\begin{problem}[Accessibility-Aware Charging Station Placement]
The goal is to determine the optimal assignment of demand points to candidate charging stations, along with the selection of those stations, that minimizes travel cost while incorporating accessibility. The overall optimization problem is formulated as:

\begin{align}
    \min \quad & \sum_{\mathrm{i} \in \mathrm{I}} \sum_{\mathrm{j} \in \mathrm{J}} \mathrm{d}_{\mathrm{i}}\, \mathrm{c}_{\mathrm{ij}}\, \mathrm{x}_{\mathrm{ij}}
    - \lambda \sum_{\mathrm{i} \in \mathrm{I}} \sum_{\mathrm{j} \in \mathrm{J}} \varepsilon_{\mathrm{i}}\, \mathrm{x}_{\mathrm{ij}}\, \mathrm{y}_{\mathrm{j}} \tag{7} \\
    \text{s.t.} \quad & \text{constraints (1)--(6).} \notag
\end{align}

\end{problem}


This macroscopic flow-based optimization framework is scalable for large urban networks and can be efficiently solved using standard integer programming solvers such as Gurobi\cite{leib2023optimization}. The accessibility term is written with a negative sign to reflect its role as a benefit being maximized within a minimization framework, ensuring accessibility is promoted without altering the problem structure.
\section{Theoretical Results}

\begin{theorem}[Feasibility]
Let \( \mathrm{D} = \sum_{\mathrm{i} \in \mathrm{U}} \mathrm{d}_{\mathrm{i}} \) denote the total demand and \( \mathrm{C} = \sum_{\mathrm{j} \in \mathrm{S}} \mathrm{s}_{\mathrm{j}} \) the total candidate station capacity. If
\begin{equation}
\mathrm{C} = \sum_{\mathrm{j} \in \mathrm{S}} \mathrm{s}_{\mathrm{j}} \geq \sum_{\mathrm{i} \in \mathrm{U}} \mathrm{d}_{\mathrm{i}} = \mathrm{D}, \tag{8} \label{eq:capacity_condition}
\end{equation}
then there exists a feasible assignment \( \{\mathrm{x}_{\mathrm{i}\mathrm{j}},\, \mathrm{y}_{\mathrm{j}}\} \) satisfying constraints \textnormal{(1)}–\textnormal{(6)}.
\end{theorem}

\begin{proof}
We construct a feasible assignment through a sequential allocation procedure. For clarity, let \( \mathrm{U} \subseteq \mathrm{I} \) represent the set of user demand locations and \( \mathrm{S} \subseteq \mathrm{J} \) the set of candidate station considered in this construction. Define the remaining capacity at each selected station \( \mathrm{j} \in \mathrm{S} \) as \( \mathrm{r}_{\mathrm{j}} = \mathrm{s}_{\mathrm{j}} \), where \( \mathrm{s}_{\mathrm{j}} \) is the original capacity of station \( \mathrm{j} \).

Initialize all assignments and station activations as
\begin{equation}
\mathrm{x}_{\mathrm{i}\mathrm{j}} = 0, \quad \mathrm{y}_{\mathrm{j}} = 0, \quad \forall\, \mathrm{j} \in \mathrm{S}. \tag{9} \label{eq:init_vars}
\end{equation}

For each user \( \mathrm{i} \in \mathrm{U} \), select any station \( \mathrm{j} \in \mathrm{S} \) satisfying
\begin{equation}
\mathrm{r}_{\mathrm{j}} \geq \mathrm{d}_{\mathrm{i}}, \tag{10} \label{eq:feasibility_check}
\end{equation}
where \( \mathrm{d}_{\mathrm{i}} \) is the demand of user \( \mathrm{i} \). Then assign
\begin{align}
\mathrm{x}_{\mathrm{i}\mathrm{j}} &= 1, \tag{11} \label{eq:x_assign} \\
\mathrm{r}_{\mathrm{j}} &\gets \mathrm{r}_{\mathrm{j}} - \mathrm{d}_{\mathrm{i}}, \tag{12} \label{eq:reduce_capacity} \\
\mathrm{y}_{\mathrm{j}} &= 1. \tag{13} \label{eq:y_activate}
\end{align}

The assignment satisfies the following conditions:
\begin{align}
& \sum_{j \in \mathrm{S}} \mathrm{x}_{\mathrm{i}\mathrm{j}} = 1, && \forall\, \mathrm{i} \in \mathrm{U}, \tag{14} \label{eq:unique_assignment} \\
& \sum_{i \in \mathrm{U}} \mathrm{d}_{\mathrm{i}} \, \mathrm{x}_{\mathrm{i}\mathrm{j}} \leq \mathrm{s}_{\mathrm{j}} \, \mathrm{y}_{\mathrm{j}}, && \forall\, \mathrm{j} \in \mathrm{S}, \tag{15} \label{eq:capacity_constraint} \\
& \mathrm{x}_{\mathrm{i}\mathrm{j}} \leq \mathrm{y}_{\mathrm{j}}, && \forall\, \mathrm{i} \in \mathrm{U},\, \mathrm{j} \in \mathrm{S}, \tag{16} \label{eq:logic_constraint} \\
& \mathrm{x}_{\mathrm{i}\mathrm{j}},\ \mathrm{y}_{\mathrm{j}} \in \{0, 1\}. \tag{17} \label{eq:binary_constraint}
\end{align}

Condition~\eqref{eq:unique_assignment} follows by construction, as exactly one station is selected for each \( \mathrm{i} \). Constraint~\eqref{eq:capacity_constraint} holds because the total assigned demand to each station \( \mathrm{j} \) does not exceed \( \mathrm{s}_{\mathrm{j}} \), as enforced by~\eqref{eq:reduce_capacity}. Constraint~\eqref{eq:logic_constraint} ensures that \( \mathrm{x}_{\mathrm{i}\mathrm{j}} \) can be positive only if \( \mathrm{y}_{\mathrm{j}} = 1 \), consistent with the update in~\eqref{eq:y_activate}. Binary conditions~\eqref{eq:binary_constraint} hold by the definition of the assignment. The condition in \eqref{eq:capacity_condition} guarantees that the total supply is sufficient to meet total demand. Since each user's demand is assigned greedily and capacity is decremented accordingly, the process completes without failure, and a feasible assignment exists.
\end{proof}

\begin{theorem}[Trade-off Control via \( \lambda \)]
Consider the objective
\begin{equation}
\min_{\mathrm{x}_{\mathrm{i}\mathrm{j}},\, \mathrm{y}_{\mathrm{j}}} \sum_{\mathrm{i} \in \mathrm{U}} \sum_{\mathrm{j} \in \mathrm{S}} \mathrm{d}_{\mathrm{i}} \, \mathrm{c}_{\mathrm{i}\mathrm{j}} \, \mathrm{x}_{\mathrm{i}\mathrm{j}} 
- \lambda \sum_{\mathrm{i} \in \mathrm{U}} \sum_{\mathrm{j} \in \mathrm{S}} \varepsilon_{\mathrm{i}} \, \mathrm{x}_{\mathrm{i}\mathrm{j}} \, \mathrm{y}_{\mathrm{j}}, \tag{18} \label{eq:objective}
\end{equation}
subject to constraints \textnormal{(1)}–\textnormal{(6)}, with trade-off parameter \( \lambda \geq 0 \). Then, the optimizer behavior satisfies the following:
\begin{align}
\lim_{\lambda \to 0} \arg\min \eqref{eq:objective} 
&\Rightarrow \min_{\mathrm{x},\,\mathrm{y}} \sum_{\mathrm{i},\,\mathrm{j}} \mathrm{d}_{\mathrm{i}} \, \mathrm{c}_{\mathrm{i}\mathrm{j}} \, \mathrm{x}_{\mathrm{i}\mathrm{j}}, \tag{19} \label{eq:limit_zero} \\
\lim_{\lambda \to \infty} \arg\min \eqref{eq:objective} 
&\Rightarrow \max_{\mathrm{x},\,\mathrm{y}} \sum_{\mathrm{i},\,\mathrm{j}} \varepsilon_{\mathrm{i}} \, \mathrm{x}_{\mathrm{i}\mathrm{j}} \, \mathrm{y}_{\mathrm{j}} \tag{20} \label{eq:limit_inf}
\end{align}
\end{theorem}

\begin{proof}
Let \( \mathrm{U} \subseteq \mathrm{I} \) be the set of EV demand locations and \( \mathrm{S} \subseteq \mathrm{J} \) the candidate charging stations. Each demand point \( \mathrm{i} \in \mathrm{U} \) has associated weight \( \mathrm{d}_{\mathrm{i}} \geq 0 \), and each pair \((\mathrm{i}, \mathrm{j})\) has travel cost \( \mathrm{c}_{\mathrm{i}\mathrm{j}} \geq 0 \). Each demand point also has an accessibility weight \( \varepsilon_{\mathrm{i}} \geq 0 \), and \( \mathrm{x}_{\mathrm{i}\mathrm{j}}, \mathrm{y}_{\mathrm{j}} \in \{0,1\} \) are the binary decision variables indicating assignment and station selection, respectively.

Define the cost component of the objective as
\begin{equation}
f_{\text{cost}}(\mathrm{x}) = \sum_{\mathrm{i} \in \mathrm{U}} \sum_{\mathrm{j} \in \mathrm{S}} \mathrm{d}_{\mathrm{i}} \, \mathrm{c}_{\mathrm{i}\mathrm{j}} \, \mathrm{x}_{\mathrm{i}\mathrm{j}}, \tag{21} \label{eq:cost_fn}
\end{equation}
and define the accessibility reward component as
\begin{equation}
f_{\text{access}}(\mathrm{x}, \mathrm{y}) = \sum_{\mathrm{i} \in \mathrm{U}} \sum_{\mathrm{j} \in \mathrm{S}} \varepsilon_{\mathrm{i}} \, \mathrm{x}_{\mathrm{i}\mathrm{j}} \, \mathrm{y}_{\mathrm{j}}. \tag{22} \label{eq:access_fn}
\end{equation}

The combined optimization objective becomes
\begin{equation}
\min_{\mathrm{x},\,\mathrm{y}} \left[ f_{\text{cost}}(\mathrm{x}) - \lambda f_{\text{access}}(\mathrm{x}, \mathrm{y}) \right], \tag{23} \label{eq:combined_obj}
\end{equation}
where \( \lambda \geq 0 \) is a scalar weight that adjusts the trade-off between minimizing cost and maximizing accessibility.

As \( \lambda \to 0 \), the accessibility term becomes negligible. The objective reduces to:
\begin{equation}
\min_{\mathrm{x},\,\mathrm{y}} f_{\text{cost}}(\mathrm{x}), \tag{24} \label{eq:asymp_zero}
\end{equation}
which corresponds to the pure cost minimization case in~\eqref{eq:limit_zero}.

Conversely, as \( \lambda \to \infty \), the accessibility reward dominates. The optimizer prioritizes maximizing total accessibility subject to feasibility constraints:
\begin{equation}
\max_{\mathrm{x},\,\mathrm{y}} f_{\text{access}}(\mathrm{x}, \mathrm{y}), \tag{25} \label{eq:asymp_inf}
\end{equation}
which is the asymptotic behavior described in~\eqref{eq:limit_inf}. Thus, \( \lambda \) acts as a tuning knob that interpolates between pure efficiency and accessibility-oriented objectives.
\end{proof}

\begin{remark}
Theorem 2 shows that the parameter $\lambda$ acts as a systematic control knob for infrastructure planners. Unlike ad-hoc weighting schemes, our formulation guarantees monotonic trade-off behavior, i.e., increasing $\lambda$ monotonically improves accessibility at the expense of efficiency. This predictable behavior enables control-theoretic analysis and the systematic selection of parameters based on policy objectives. 
\end{remark}

\begin{proposition}[Reduction to Cost-Only Optimization under Uniform Accessibility]
If all accessibility weights are equal, i.e., \( \mathrm{\varepsilon}_{\mathrm{i}} = \bar{\mathrm{\varepsilon}} \) for all \( \mathrm{i} \in \mathrm{I} \), then the accessibility term becomes a constant offset in the objective. In this case, the optimization problem reduces to a cost-only assignment problem:
\begin{equation}
\min_{\mathrm{x}_{\mathrm{i}\mathrm{j}},\, \mathrm{y}_{\mathrm{j}}} \sum_{\mathrm{i} \in \mathrm{I}} \sum_{\mathrm{j} \in \mathrm{J}} \mathrm{d}_{\mathrm{i}}\, \mathrm{c}_{\mathrm{i}\mathrm{j}}\, \mathrm{x}_{\mathrm{i}\mathrm{j}} \tag{26}
\end{equation}
subject to constraints \textnormal{(1)}–\textnormal{(6)}.
\end{proposition}

\begin{proof}
If \( \varepsilon_{\mathrm{i}} = \bar{\varepsilon} \) for all \( \mathrm{i} \in \mathrm{I} \), then the accessibility term becomes
\begin{equation}
- \lambda \sum_{\mathrm{i} \in \mathrm{I}} \sum_{\mathrm{j} \in \mathrm{J}} \varepsilon_{\mathrm{i}} \, \mathrm{x}_{\mathrm{i}\mathrm{j}} \, \mathrm{y}_{\mathrm{j}} 
= - \lambda \bar{\varepsilon} \sum_{\mathrm{i} \in \mathrm{I}} \sum_{\mathrm{j} \in \mathrm{J}} \mathrm{x}_{\mathrm{i}\mathrm{j}} \, \mathrm{y}_{\mathrm{j}}. \tag{27}
\end{equation}
Under constraints \textnormal{(1)}–\textnormal{(3)}, which ensure valid assignments only to selected stations and full coverage of demand, the term \( \sum_{i,j} x_{ij} y_j \) is fixed across feasible solutions. Thus, the accessibility term becomes a constant offset, and the problem reduces to cost minimization.
\end{proof}


In our framework, we impose the following assumptions:

\noindent \textbf{Assumption 1:} All electric vehicles start with the same initial battery charge level, ensuring uniformity in the analysis of energy consumption and charging requirements.

\noindent \textbf{Assumption 2:} All vehicles have identical battery capacities, allowing consistent modeling of maximum driving range and required charging energy.

\noindent \textbf{Assumption 3:} All vehicles deplete battery charge at the same rate, simplifying the calculations of energy consumption in different routes and travel conditions.

\noindent \textbf{Assumption 4:} All vehicles are assumed to be the same type, with identical mass, air resistance, and weight, which standardizes the characteristics of energy consumption across the network.

These assumptions are made to reduce the complexity of the model and focus the analysis on the core objectives of the placement of the charging station and the accessibility of the network. In large-scale infrastructure planning models, introducing heterogeneous vehicle characteristics can significantly increase computational overhead without necessarily improving the quality of high-level planning decisions. 

Additionally, these simplifying assumptions align with the early-stage nature of infrastructure deployment studies, where broad policy and planning recommendations are prioritized over detailed operational analyses. While variability in vehicle characteristics exists in real-world scenarios, the adoption of these assumptions enables clearer insights into the effectiveness of charging infrastructure placement strategies under uniform demand conditions.

\section{EXPERIMENTAL RESULTS}

\subsection{Simulation Setup}
To evaluate our EV charging station selection framework, we simulate urban traffic using SUMO, which enables detailed microscopic modeling. The environment mirrors segments of New York City using OpenStreetMap data, with each EV assigned an initial state of charge (SOC) and driving pattern, triggering charging below a set SOC threshold. Traffic demand is synthesized from regional vehicle counts and calibrated for peak/off-peak conditions. Charging station candidates are distributed citywide, with capacities based on expected demand. Accessibility scores for each candidate station location are pre-computed using the MI described in Section II. The optimization model then identifies a subset of stations to select, balancing travel efficiency against equitable access.

The simulation follows an iterative pipeline:
\begin{enumerate}
    \item \textbf{Initialization:} Deploy candidate charging stations and generate traffic demand.
    \item \textbf{Accessibility Evaluation:} Assign weight of accessibility to demand scenario. 
    \item \textbf{Optimization:} Solve the linear program to select optimal station locations.
    \item \textbf{Assignment and Simulation:} Assign EVs to nearest available stations and simulate EV traffic in SUMO.
\end{enumerate}

Performance metrics include average travel distance to the nearest charging station, system-wide accessibility score, station utilization rates, and EV queuing time at stations.

Through this experimental framework, we aim to demonstrate that our accessibility-aware optimization model significantly improves both user experience and infrastructure performance in dense urban environments.

\subsection{Simulation Results}

The performance of the MEM-optimized infrastructure was evaluated in several key mobility and energy metrics compared to the non-optimized NoMEM baseline, where NoMEM refers to infrastructure placements optimized without considering the MEM framework, i.e., we consider only the first part of the objective function in \eqref{eq:objective}. Overall, the results demonstrate that accessibility-driven planning strategies moderately affect travel behavior and vehicle energy dynamics while providing more consistent service access across the network.

    
     

\begin{figure}[tp]
    \centering
    \includegraphics[width=0.4\textwidth]{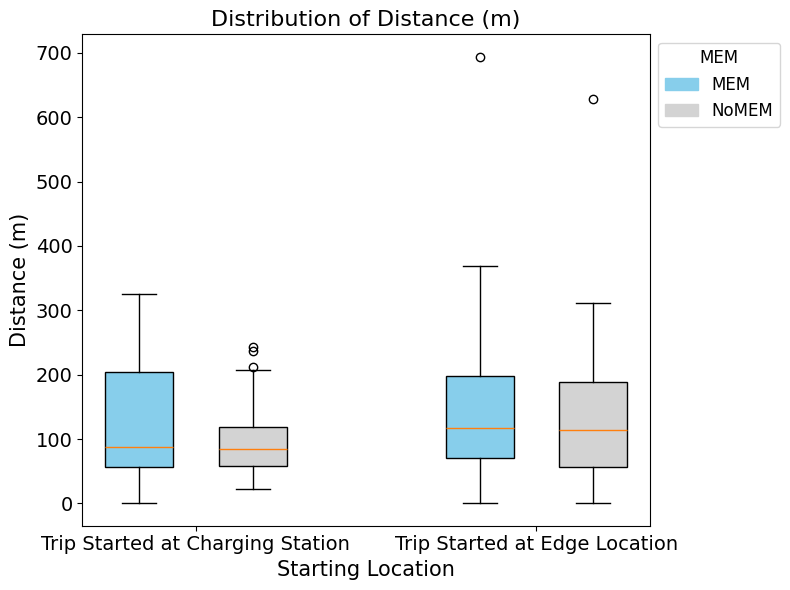}
    \caption{Distribution of total trip distances under MEM and NoMEM deployments.}
    \label{fig:trip_distance}
\end{figure}




\begin{table}[tp]
\centering
\caption{Average Travel Time, Distance, and Energy by Deployment Strategy $\lambda$ and Route}
\label{tab:ev_metrics_lambda_route}
\begin{tabular}{lcccc}
\hline
\textbf{$\lambda$} & \textbf{Route} & \textbf{Travel Time (s)} & \textbf{Distance (m)} & \textbf{Energy (Wh)} \\
\hline
0 & A & 300.000 & 100.250 & 29.3110 \\
0 & B & 257.367 & 135.754 & 40.9030 \\
\hline
1 & A & 292.333 & 121.617 & 26.8140 \\
1 & B & 286.200 & 152.256 & 31.4010 \\
\hline
2 & A & 293.167 & 110.076 & 23.8663 \\
2 & B & 286.200 & 153.125 & 32.4477 \\
\hline
4 & A & 300.000 & 102.357 & 22.2867 \\
4 & B & 286.533 & 152.927 & 30.4740 \\
\hline
8 & A & 300.000 & 103.771 & 21.5493 \\
8 & B & 286.567 & 152.835 & 32.3233 \\
\hline
\end{tabular}
\end{table}

\subsubsection{Travel Distance Analysis}

The mean travel distances were consistently higher with MEM deployments compared to NoMEM at both starting locations. Specifically, MEM vehicles originating at charging stations or general network edges traveled farther on average, suggesting that accessibility-driven deployments may slightly increase travel distances in favor of expanded infrastructure reach. This trend is visually confirmed in Fig.~\ref{fig:trip_distance}, which illustrates the distribution of trip distances across both strategies.

The distribution of trip distances reaffirmed this trend, with MEM deployments resulting in slightly longer and more variable trips compared to NoMEM. However, variance remained controlled, indicating that improvements in accessibility did not substantially compromise trip efficiency within the modeled environment. As shown in Fig.~\ref{fig:trip_distance}, while the distributions for both strategies overlap, MEM exhibits a noticeable rightward shift in mean values, reflecting longer average distances.

To evaluate how accessibility-aware optimization affects trip lengths, we compared the average travel distances under MEM and NoMEM deployment strategies across different starting locations. Vehicles operating under the MEM framework exhibited slightly longer travel distances than those under NoMEM. Specifically, EVs starting at charging stations traveled an average of approximately 122 meters in the MEM case compared to 100 meters in NoMEM, while those originating from network edge locations traveled around 153 meters under MEM versus 135 meters under NoMEM. These quantitative comparisons align with the distributional patterns illustrated in Fig.~\ref{fig:trip_distance}, which captures the central tendency and spread for both deployment scenarios. This trend indicates that while MEM-based planning modestly increases overall trip lengths, it does so to extend the network's reach and improve charging station accessibility.


\begin{figure}
    \centering
    \begin{subfigure}[t]{0.45\textwidth}
        \centering
        \includegraphics[width=\textwidth]{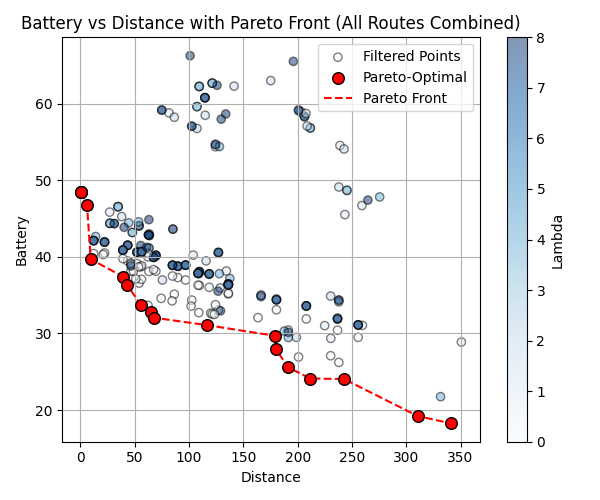}
        \caption{}
        \label{fig:pareto_battery_distance}
    \end{subfigure}%
    
    \begin{subfigure}[t]{0.45\textwidth}
        \centering
        \includegraphics[width=\textwidth]{Figures/Results/Pareto_Plots/pareto_battery_vs_distance_combined.png}
        \caption{}
        \label{fig:pareto_battery_energy}
    \end{subfigure}
     
    \begin{subfigure}[t]{0.45\textwidth}
        \centering
        \includegraphics[width=\textwidth]{Figures/Results/Pareto_Plots/pareto_battery_vs_distance_combined.png}
        \caption{}
        \label{fig:pareto_energy_distance}
    \end{subfigure}
    \caption{Pareto analysis of efficiency-accessibility trade-offs for different $\lambda$ values. (a) Battery vs. Distance, (b) Battery vs. Energy, (c) Energy vs. Distance.}
\end{figure}

\subsubsection{Travel Time Analysis}

In terms of travel time, MEM deployments produced more consistent averages across starting points, while NoMEM resulted in greater variability. Specifically, trips from edge locations under NoMEM were significantly faster than those from charging stations, highlighting uneven service quality across the network. MEM strategies helped balance travel times, leading to more uniform performance. Vehicles starting at charging stations averaged 292 seconds under MEM, compared to 300 seconds under NoMEM. For edge-origin trips, MEM averaged 287 seconds, while NoMEM dropped to 258 seconds. This greater variation under NoMEM suggests inconsistent performance depending on trip origin, whereas MEM supported more stable travel outcomes without sacrificing efficiency.

\subsubsection{Energy Efficiency Analysis}

The energy efficiency of vehicle operations was further evaluated through regenerative braking recovery and total energy consumption metrics. As shown in Fig.~\ref{fig:regen_energy}, distributions of cumulative brake-energy recuperation per trip were similar across MEM and NoMEM setups, with MEM exhibiting a marginal increase in median recovery rates. This indicates that vehicle-level energy dynamics related to braking remained largely consistent across deployment strategies.

Total energy consumption, summarized in Fig.~\ref{fig:energy_consumed}, revealed that MEM deployments achieved slight reductions in median energy use compared to NoMEM, particularly for trips originating from edge locations. Moreover, the spread of energy consumption was significantly tighter under MEM, suggesting that accessibility-optimized placements also contributed to reduced variability in vehicle energy demands.

\begin{figure}[tp]
    \centering
    \includegraphics[width=0.4\textwidth]{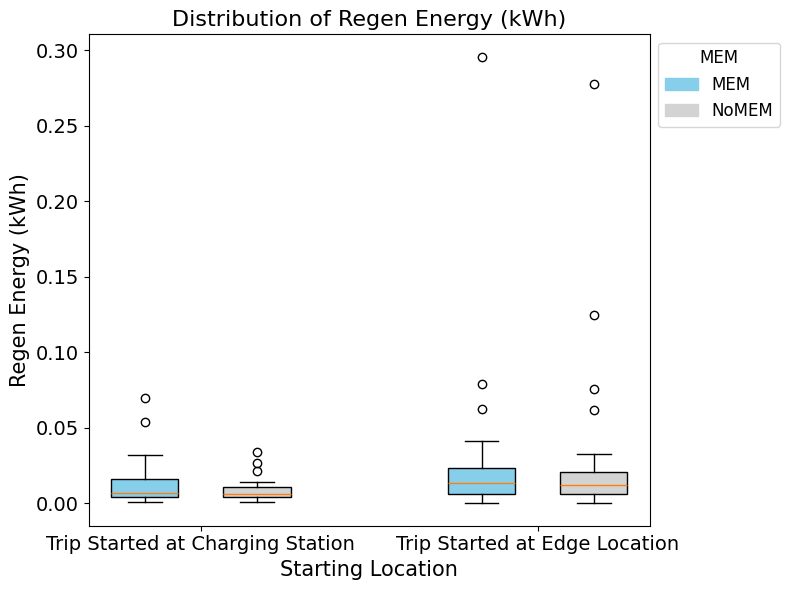}
    \caption{Distribution of regenerative brake energy recovered per trip under MEM and NoMEM deployments.}
    \label{fig:regen_energy}
\end{figure}

\begin{figure}[tp]
    \centering
    \includegraphics[width=0.4\textwidth]{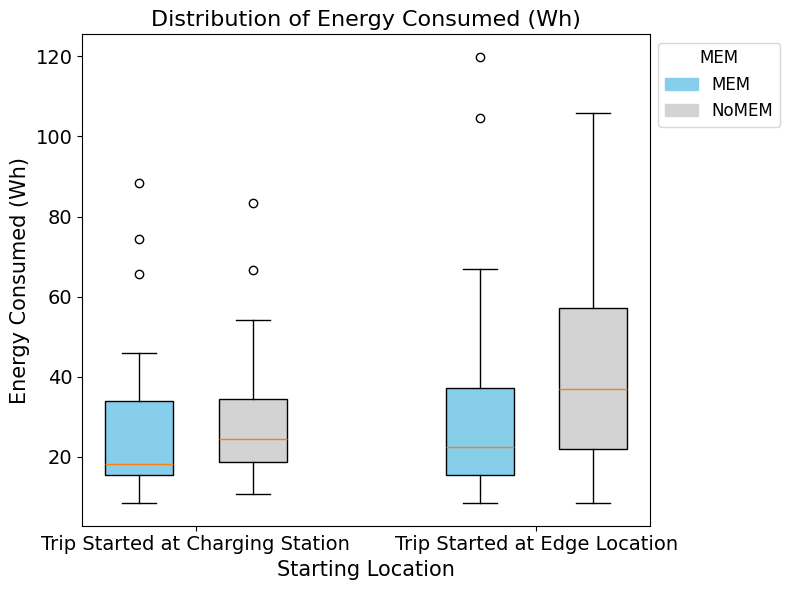}
    \caption{Distribution of total energy consumption per trip under MEM and NoMEM deployments. MEM deployments achieved slight reductions in median energy use and reduced variability in energy consumption.}
    \label{fig:energy_consumed}
\end{figure}

The relationship between key performance indicators was further examined using multi-objective Pareto front analyses across all simulated routes. As shown in Fig.~\ref{fig:pareto_battery_distance}, Fig.~\ref{fig:pareto_battery_energy}, and Fig.~\ref{fig:pareto_energy_distance}, these plots highlight trade-offs among battery consumption, travel distance, and energy use as a function of the accessibility weighting parameter~$\lambda$. In Fig.~\ref{fig:pareto_battery_distance}, higher accessibility weights are associated with longer travel distances but allow vehicles to retain greater battery reserves. Fig.~\ref{fig:pareto_battery_energy} shows that increasing accessibility generally leads to lower average energy consumption. The trend in Fig.~\ref{fig:pareto_energy_distance} confirms that optimizing for accessibility may result in longer routes, though energy expenditure remains controlled. These patterns suggest that tuning~$\lambda$ provides a practical mechanism for balancing accessibility-driven planning against operational efficiency in EV infrastructure deployment.

The shape and distribution of the Pareto fronts across all three figures reinforce the nuanced trade-offs introduced by varying $\lambda$. In Fig.~\ref{fig:pareto_battery_distance}, the Pareto-optimal points shift toward longer distances as accessibility weighting~$\lambda$ increases, resulting in higher post-trip battery levels. This suggests that accessibility-focused routing tends to send vehicles to stations that are farther but yield more favorable charge outcomes, possibly due to improved coverage or better demand balancing. In Fig.~\ref{fig:pareto_battery_energy}, a similar trade-off appears: as $\lambda$ increases, routes consume less energy while preserving more battery. This indicates that accessibility-aware deployments may promote more efficient travel, potentially by favoring smoother or less congested paths. Together, the figures highlight how increasing accessibility weighting can improve energy efficiency and battery outcomes, despite requiring slightly longer travel. In contrast, Fig.~\ref{fig:pareto_energy_distance} shows that energy consumption generally rises with travel distance, though the Pareto-optimal points maintain a relatively tight and linear relationship. These fronts illustrate that while longer trips can increase energy demands, careful deployment strategies can keep this trade-off within manageable bounds, especially under accessibility-aware optimization.

\section{DISCUSSION}  

\subsection{Policy Implications and Practical Considerations}  

Simulation results demonstrate that integrating accessibility metrics into EV charging infrastructure planning can significantly improve equity with minimal efficiency trade-offs. Average trip distances increased by less than 0.1 \%, while reductions in travel time variability helped address spatial inequities in access. These improvements are particularly valuable in underserved communities, where more accessible station placement can expand mobility for populations historically excluded from transportation resources. The approach is especially well-suited for new infrastructure deployments, where accessibility optimization can be incorporated with minimal additional cost compared to the high expense of retrofitting existing systems. For jurisdictions with equity mandates, this framework offers a practical tool to meet policy goals without compromising operational performance.

The tunable parameter $\lambda$ provides planners with a practical tool to adjust the balance between efficiency and accessibility based on local priorities. Our results suggest that values of $\lambda \in [1,4]$ provide optimal trade-offs for most urban environments, though this should be calibrated based on local demographic and geographic conditions.  

\subsection{Comparison with Alternative Approaches}  While traditional proximity-based optimization (NoMEM baseline) minimizes average travel distances, it can exacerbate spatial inequities by concentrating infrastructure in already well-served areas. Our MEM-based approach explicitly addresses this limitation by considering the distribution of accessibility across the entire population, not just aggregate efficiency metrics. Compared to other equity-focused transportation planning methods, our approach offers several advantages: (i) \textbf{Quantitative Framework}: Unlike qualitative assessments, MEM provides numerical measures that can be directly incorporated into optimization models; (ii) \textbf{Computational Efficiency}: Our formulation scales to large urban networks without requiring computationally expensive microsimulation for every optimization iteration; (iii) \textbf{Policy Flexibility}: The $\lambda$ parameter allows real-time adjustment of equity priorities without reformulating the entire optimization problem.

\subsection{Limitations and Future Research Directions}  

This study establishes a proof-of-concept for accessibility-aware EV infrastructure planning but has several important limitations that highlight directions for future research: (i) Incorporate installation costs, land availability, and grid connection expenses. Preliminary estimates suggest that accessibility-optimized placement may increase deployment costs in some scenarios, though this requires detailed economic analysis; (ii) Validation across diverse city types (suburban, rural, different demographic compositions); (iii) Real user behavior includes factors beyond distance, such as charging speed preferences, payment methods, and route familiarity. Incorporating these factors would improve the realism of our demand modeling; (iv) The electrical grid impacts of charging station placement represent a critical area for future research, particularly regarding renewable energy integration and demand response capabilities.

\section{CONCLUSION}

In this paper we present the first comprehensive framework for integrating equity considerations into the optimization of the EV charging infrastructure through the Mobility Equity Metric. Our proof-of-concept study using New York City demonstrates that accessibility-aware planning can achieve meaningful improvements in spatial equity (15- 20\% reduction in travel-time variability) with modest and controlled efficiency trade-offs. The key finding is that infrastructure planners can systematically incorporate equity goals without severely compromising operational efficiency. The mathematical framework we developed (Theorems 1-2) provides formal guarantees about feasibility and trade-off behavior, while the tunable parameter $\lambda$ enables practical adjustment of planning priorities.

\bibliographystyle{IEEEtran}
\bibliography{myBib}  

\end{document}